\documentclass[10pt,conference]{IEEEtran}

\usepackage[english]{babel} 
\usepackage{amsmath}                
\usepackage{amsfonts}               
\usepackage{amssymb}                
\usepackage{amsopn}                 
\usepackage{bbm}                    
\usepackage{mathrsfs}               
\usepackage{calc}                   
\usepackage[dvips]{graphicx}        
\usepackage{epsfig}                
\usepackage{psfrag}                 
\usepackage[dvips]{color}           
\usepackage{fancyhdr}               
\usepackage{verbatim}               
\usepackage{exscale}                
\usepackage{macros}
\newtheorem{theorem}{Theorem}

\newtheorem{proposition}[theorem]{Proposition}

\title{Asymmetric Quantizers Are Better at Low SNR}

\author{\IEEEauthorblockN{Tobias Koch}\IEEEauthorblockA{University of Cambridge\\ Cambridge CB2 1PZ, UK\\ Email: tobi.koch@eng.cam.ac.uk} \and \IEEEauthorblockN{Amos Lapidoth}\IEEEauthorblockA{ETH Zurich\\ 8092 Zurich, Switzerland\\
    Email: lapidoth@isi.ee.ethz.ch}}

\date{}

\begin{document}

\maketitle

\begin{abstract}
  We study the behavior of channel capacity when a one-bit quantizer
  is employed at the output of the discrete-time average-power-limited Gaussian channel. We focus on the low signal-to-noise ratio regime, where communication at very low spectral
  efficiencies takes place, as in Spread-Spectrum and Ultra-Wideband communications. 
  It is well known that, in this regime, a
  symmetric one-bit quantizer reduces capacity by $2/\pi$, which
  translates to a power loss of approximately two decibels. Here we
  show that if an asymmetric one-bit quantizer is employed, and if
  asymmetric signal constellations are used, then these two
  decibels can be recovered in full. 
\renewcommand{\thefootnote}{}
  \footnote{The research leading to these results has received funding
    from the European Community's Seventh Framework Programme
    (FP7/2007-2013) under grant agreement No. 252663.}
\end{abstract}
\setcounter{footnote}{0}

\section{Introduction}
\label{sec:intro}
We study the effect on channel capacity of quantizing the output of
the discrete-time average-power-limited Gaussian channel using a
one-bit quantizer. We focus on the low signal-to-noise ratio
regime, where communication at very low spectral efficiencies takes place (as in
Spread-Spectrum and Ultra-Wideband communications).
%
%
%
In this regime, a symmetric one-bit quantizer reduces the capacity by
a factor of $2/\pi$, corresponding to a 2dB power loss
\cite{viterbiomura79}. Hence the rule of thumb that ``hard decisions
cause a 2dB power loss.'' Here we demonstrate that if we allow for
\emph{asymmetric one-bit quantizers} with corresponding
\emph{asymmetric signal constellations}, these two decibels can be
fully recovered. We further demonstrate that the capacity per
unit-energy can be achieved by a simple pulse-position modulation
(PPM) scheme.

The problem of output quantization is relevant for communication
systems where the receiver uses digital signal processing techniques,
which require the conversion of the analog received signal to a
digital signal by means of an analog-to-digital converter (ADC). For ADCs with high resolution, the effects of quantization are negligible. However, using a high-resolution ADC may not be practical, especially when the bandwidth of the communication system is large and the ADC therefore needs to operate at a high sampling rate \cite{walden99}. In this case a low-resolution ADC has to be employed. The capacity of the discrete-time Gaussian channel with one-bit output quantization indicates what communication rates can be achieved when the receiver employs a low-resolution ADC.

For a symmetric one-bit quantizer (which produces $1$ if the quantizer input is nonnegative and $-1$ otherwise), the capacity $C_{\textnormal{sym}}(\const{P})$ under the average-power constraint $\const{P}$ on the channel inputs is given by \cite[(3.4.18)]{viterbiomura79}, \cite[Th.~2]{singhdabeermadhow09_2}
\begin{equation}
\label{eq:sym}
C_{\textnormal{sym}}(\const{P}) = \log 2 - H_b\biggl(Q\Bigl(\sqrt{\const{P}/\sigma^2}\Bigr)\biggr)
\end{equation}
where $\log(\cdot)$ denotes the natural logarithm function; $\sigma^2$ the variance of the additive Gaussian noise; 
\begin{equation*}
H_b(p)\triangleq-p\log p - (1-p)\log(1-p), \quad 0 \leq p \leq 1 
\end{equation*}
(with $0\log 0\triangleq 0$) the binary entropy function; and 
\begin{equation*}
Q(x)\triangleq \frac{1}{\sqrt{2\pi}}\int_{x}^{\infty} \exp\left(-\xi^2/2\right)\d \xi, \quad x\in\Reals
\end{equation*}
the $Q$-function. (Here $\Reals$ denotes the set of real numbers.) This capacity can be achieved by transmitting $\pm\sqrt{\const{P}}$ equiprobably.

From \eqref{eq:sym}, the capacity per unit-energy $\dot{C}_{\textnormal{sym}}(0)$ for a symmetric one-bit quantizer can be computed as \cite[(3.4.20)]{viterbiomura79}
\begin{equation}
\dot{C}_{\textnormal{sym}}(0) = \lim_{\const{P}\downarrow 0} \frac{C_{\textnormal{sym}}(\const{P})}{\const{P}} = \frac{1}{\pi\sigma^2}.
\end{equation}
This is a factor of $\frac{2}{\pi}$ smaller than the capacity per-unit energy $\frac{1}{2\sigma^2}$ of the Gaussian channel without output quantization \cite{shannon48}. Thus, quantizing the channel output by a symmetric one-bit quantizer causes a loss of roughly 2dB. 


It is tempting to attribute this loss to the fact that the quantizer discards information on the received signal's magnitude and allows the decoder to perform only hard-decision decoding. However, we demonstrate that the loss of 2dB is not a consequence of the hard-decision decoder but of the suboptimal quantizer. In fact, with an asymmetric quantizer the loss vanishes.

The rest of the paper is organized as follows. Section~\ref{sec:channel} describes the considered channel model, defines the capacity per unit-energy, and presents our main result. Section~\ref{sec:mainproof} provides the proof of this result. Section~\ref{sec:ppm} shows that the capacity per unit-energy can be achieved by a simple PPM scheme. Section~\ref{sec:capacity} briefly discusses the capacity of the considered channel. And Section~\ref{sec:summary} concludes the paper with a summary and discussion of our results.

\section{Channel Model and Capacity Per Unit-Energy}
\label{sec:channel}
\begin{figure*}[t]
\centering
\psfrag{X}[cb][cb]{$X_k$}
\psfrag{M}[rc][rc]{$M$}
 \psfrag{Z}[cc][cc]{$Z_k$}
\psfrag{encoder}[cc][cc]{encoder}
\psfrag{decoder}[cc][cc]{decoder}
\psfrag{quantizer}[cc][cc]{quantizer}
\psfrag{Y}[cb][cb]{$Y_k$}
\psfrag{N}[lc][lc]{$\hat{M}$}
\psfrag{U}[cb][cb]{$\tilde{Y}_k$}
\epsfig{file=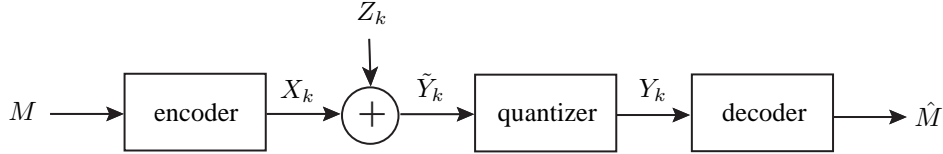, width=0.7\textwidth}
 \caption{System model.}
 \label{fig1}
\end{figure*}

We consider the discrete-time communication system depicted in Figure 1. It is assumed that the message $M$ is uniformly distributed over the set $\{1,2,\ldots,\const{M}\}$. The encoder maps $M$ to the length-$n$ sequence $X_1,X_2,\ldots,X_n$, which is then corrupted by additive Gaussian noise and quantized. Thus, at every time instant $k\in\Integers$ (where $\Integers$ denotes the set of integers), the received signal $\tilde{Y}_k$ corresponding to the channel input $x_k\in\Reals$ is
\begin{equation}
\label{eq:channel}
\tilde{Y}_k = x_k + Z_k, \quad k\in\Integers
\end{equation}
where  $\{Z_k,\,k\in\Integers\}$ is a sequence of independent and identically distributed (i.i.d.) Gaussian random variables of zero mean and variance $\sigma^2$. The quantizer produces $1$ if $\tilde{Y}_k$ is in the quantization region $\set{D}$ and $-1$ otherwise, for some Borel set $\set{D}\subset\Reals$. For example, the quantization region of the symmetric quantizer is given by $\set{D} = \{\tilde{y}\in\Reals\colon \tilde{y}\geq 0\}$. The decoder observes the quantizer's outputs $Y_1,Y_2,\ldots,Y_n$ and tries to guess which message was transmitted.

We assume that the energy in the transmitted sequence does not exceed
$\const{E}$, i.e., that the encoder is such that, for every realization of
$M$, the sequence $x_1,x_2,\ldots,x_n$ satisfies
\begin{equation}
\label{eq:power}
\sum_{k=1}^n x^2_k \leq \const{E}.
\end{equation}
We define the capacity per unit-energy along the lines of
\cite{verdu90}. We say that a \emph{rate per unit-energy} $\dot{R}(0)$
(in nats per energy) is \emph{achievable} if for every $\eps>0$ there
exists an encoder, a quantizer, and a decoder such that
\begin{equation*}
\frac{\log\const{M}}{\const{E}} > \dot{R}(0)-\eps
\end{equation*}
and such that the probability of error $\Prob(\hat{M}\neq M)$ tends to zero as $\const{E}$ tends to infinity. The \emph{capacity per unit-energy} is the supremum of all achievable rates per unit-energy.

It follows from \cite[Ths.~2 \& 3]{verdu90} that the capacity per unit-energy is given by\footnote{It is straightforward to incorporate the additional maximization over all possible quantizers into the proofs of \cite[Ths.~2 \& 3]{verdu90}.}
\begin{IEEEeqnarray}{lClCl}
\dot{C}(0) & = & \sup_{\const{P}>0} \frac{C(\const{P})}{\const{P}} & = & \sup_{\xi\neq 0, \set{D}\subset\Reals} \frac{D\bigl(P_{Y_1 | X_1=\xi} \bigm\| P_{Y_1 | X_1=0}\bigr)}{\xi^2}\IEEEeqnarraynumspace\label{eq:formulaCUC}
\end{IEEEeqnarray}
where $D(\cdot\|\cdot)$ denotes the relative entropy, i.e.,
\begin{equation*}
D(P \| Q) = \left\{\begin{array}{cl}\displaystyle \int \log\left(\frac{\d P}{\d Q}\right)\d P, \quad & \textnormal{if }P\ll Q\\ \infty, \quad &\textnormal{otherwise}\end{array}\right.
\end{equation*}
(where $P\ll Q$ indicates that $P$ is absolutely continuous with respect to $Q$); $P_{Y_1|X_1=x}$ denotes the output distribution given that the input is $x$; and $C(\const{P})$ is the capacity of the above channel under the average-power constraint $\const{P}$ on the channel inputs, i.e., \cite{gallager68}
\begin{equation}
C(\const{P}) = \sup I(X_1;Y_1)
\end{equation} 
with the maximization being over all quantization regions $\set{D}$ and over all distributions on $X_1$ satisfying $\E{X_1^2}\leq\const{P}$. Note that the first supremum in \eqref{eq:formulaCUC} is approached when $\const{P}$ tends to zero. Thus, the capacity per unit-energy is equal to the slope of the capacity-vs-power curve at zero.

By the Data Processing Inequality \cite[Th.~2.8.1]{coverthomas91}, the
capacity per unit-energy of the above channel is upper-bounded by the
capacity per unit-energy of the Gaussian channel without output
quantization, i.e., \cite{shannon48}
\begin{equation}
\label{eq:CUCUB}
\dot{C}(0) \leq \frac{1}{2\sigma^2}.
\end{equation}
In the following, we show that there exists a one-bit quantizer that achieves this upper bound.

\begin{theorem}[Main Result]
\label{thm:main}
The capacity per unit-energy of the discrete-time Gaussian channel with one-bit output quantization is
\begin{equation}
\label{eq:mainCUC}
\dot{C}(0) = \frac{1}{2\sigma^2}.
\end{equation}
Moreover, the capacity per unit-energy can be achieved by a quantization region of the form
\begin{equation}
\label{eq:quantregion}
\set{D} = \{\tilde{y}\in\Reals\colon \tilde{y}\geq\Upsilon\},\quad \textnormal{for some $\Upsilon\in\Reals$}
\end{equation}
where $\Upsilon$ depends on the distribution of the channel input. 
\end{theorem}
\begin{proof}
See Section~\ref{sec:mainproof}.
\end{proof}
We prove Theorem~\ref{thm:main} in Section~\ref{sec:mainproof}. A simple PPM scheme that achieves the capacity per unit-energy \eqref{eq:mainCUC} is presented in Section~\ref{sec:ppm}.

\section{Proof of Theorem~\ref{thm:main}}
\label{sec:mainproof}
We show that a one-bit quantizer with quantization region~$\set{D}$ of the form \eqref{eq:quantregion} achieves the rate per unit-energy
\begin{equation}
\label{eq:CUCLB}
\dot{R}(0) \geq \frac{1}{2\sigma^2}.
\end{equation}
Together with \eqref{eq:CUCUB}, this proves Theorem~\ref{thm:main}. 

To this end, we first note that, for the quantization region \eqref{eq:quantregion}, the conditional probability of the output $Y_1$ given the input $x$ is 
\begin{IEEEeqnarray}{lCl}
P\bigl(Y_1=1\bigm| X_1=x) 
& = & Q\left(\frac{\Upsilon-x}{\sigma}\right), \quad x\in\Reals
\end{IEEEeqnarray}
and $P\bigl(Y_1=-1\bigm|X_1=x\bigr)=1-P\bigl(Y_1=1\bigm| X_1=x)$. Together with \eqref{eq:formulaCUC}, this yields
\begin{IEEEeqnarray}{lCl}
\dot{R}(0) & = & \sup_{\xi\neq 0, \Upsilon\in\Reals} \left\{\frac{Q\left(\frac{\Upsilon-\xi}{\sigma}\right)\log\frac{Q\left(\frac{\Upsilon-\xi}{\sigma}\right)}{Q\left(\frac{\Upsilon}{\sigma}\right)}}{\xi^2}\right.\nonumber\\
& & \left.\vphantom{\frac{Q\left(\frac{\Upsilon-\xi}{\sigma}\right)\log\frac{Q\left(\frac{\Upsilon-\xi}{\sigma}\right)}{Q\left(\frac{\Upsilon}{\sigma}\right)}}{\xi^2}}\!\!\!\quad\qquad\qquad {} +\frac{\left[1-Q\left(\frac{\Upsilon-\xi}{\sigma}\right)\right]\log\frac{1-Q\left(\frac{\Upsilon-\xi}{\sigma}\right)}{1-Q\left(\frac{\Upsilon}{\sigma}\right)}}{\xi^2}\right\}\nonumber\\
& = & \sup_{\xi\neq 0, \Upsilon\in\Reals} \left\{\frac{Q\left(\frac{\Upsilon-\xi}{\sigma}\right)\log\frac{1}{Q\left(\frac{\Upsilon}{\sigma}\right)}}{\xi^2}\right.\nonumber\\
& & \!\!\!\quad\qquad\qquad {} + \frac{\left[1-Q\left(\frac{\Upsilon-\xi}{\sigma}\right)\right]\log\frac{1}{1-Q\left(\frac{\Upsilon}{\sigma}\right)}}{\xi^2}\nonumber\\
& & \qquad\qquad\qquad\qquad\qquad\left.{}\vphantom{\frac{Q\left(\frac{\Upsilon-\xi}{\sigma}\right)\log\frac{1}{Q\left(\frac{\Upsilon}{\sigma}\right)}}{\xi^2}} - \frac{H_b\left(Q\left(\frac{\Upsilon-\xi}{\sigma}\right)\right)}{\xi^2}\right\}.\IEEEeqnarraynumspace\label{eq:1}
\end{IEEEeqnarray}
We choose $\Upsilon=\xi-\mu$ for some fixed $\mu\in\Reals$ and lower-bound the right-hand side (RHS) of \eqref{eq:1} by letting $\xi$ tend to infinity. This yields for the last two terms on the RHS of \eqref{eq:1}
\begin{equation}
\label{eq:2}
\lim_{\xi\to\infty} \frac{H_b\left(Q\left(-\frac{\mu}{\sigma}\right)\right)}{\xi^2} = 0
\end{equation}
and
\begin{equation}
\label{eq:3}
\lim_{\xi\to\infty} \frac{\left[1-Q\left(-\frac{\mu}{\sigma}\right)\right]\log\frac{1}{1-Q\left(\frac{\xi-\mu}{\sigma}\right)}}{\xi^2} = 0.
\end{equation}
We use the upper bound on the $Q$-function \cite[Prop.~19.4.2]{lapidoth09} 
\begin{equation}
\label{eq:QUB}
Q(\alpha) < \frac{1}{\sqrt{2\pi}\alpha}e^{-\alpha^2/2}, \quad \alpha > 0
\end{equation}
to lower-bound the first term on the RHS of \eqref{eq:1} as
\begin{IEEEeqnarray}{lCl}
\IEEEeqnarraymulticol{3}{l}{\lim_{\xi\to\infty} \frac{Q\left(-\frac{\mu}{\sigma}\right)\log\frac{1}{Q\left(\frac{\xi-\mu}{\sigma}\right)}}{\xi^2}}\nonumber\\
\quad & \geq & Q\left(-\frac{\mu}{\sigma}\right) \lim_{\xi\to\infty} \frac{\frac{1}{2}\log(2\pi)+\log\frac{\xi-\mu}{\sigma}+\frac{(\xi-\mu)^2}{2\sigma^2}}{\xi^2}\nonumber\\
& = & Q\left(-\frac{\mu}{\sigma}\right)\frac{1}{2\sigma^2}.\label{eq:4}
\end{IEEEeqnarray}
Combining \eqref{eq:2}, \eqref{eq:3}, and \eqref{eq:4} with \eqref{eq:1} yields
\begin{equation}
\dot{R}(0) \geq Q\left(-\frac{\mu}{\sigma}\right)\frac{1}{2\sigma^2}
\end{equation}
from which we obtain \eqref{eq:CUCLB} by letting $\mu$ tend to infinity. This proves Theorem~\ref{thm:main}.

\section{Pulse-Position Modulation}
\label{sec:ppm}
The capacity per unit-energy \eqref{eq:mainCUC} can be achieved by the following PPM scheme. For each message $m\in\{1,2,\ldots,\const{M}\}$, the encoder produces the sequence $x_1(m),x_2(m),\ldots,x_{\const{M}}(m)$, where
\begin{equation}
x_k(m) = \left\{\begin{array}{ll} \xi, \quad & k=m\\ 0, \quad & k\neq m \end{array}\right.
\end{equation}
and where $\xi$ satisfies the energy constraint \eqref{eq:power} with equality. Thus, we have $\xi^2=\const{E}$, which for a fixed rate per unit-energy $\dot{R}(0) = \frac{\log\const{M}}{\const{E}}$ is equal to
\begin{equation}
\label{eq:xi}
\xi^2 = \frac{\log\const{M}}{\dot{R}(0)}.
\end{equation}
Note that, while the \emph{rate per unit-energy} is fixed, the \emph{rate} of this scheme is  $\frac{\log\const{M}}{\const{M}}$ and tends to zero as $\const{M}$ tends to infinity. Nevertheless, by \eqref{eq:formulaCUC} the capacity per unit-energy is equal to the slope of the capacity-vs-power curve at zero. It thus follows from Theorem~\ref{thm:main} that there also exists a transmission scheme of nonzero rate that achieves \eqref{eq:mainCUC}.

We employ a quantizer with quantization region
\begin{equation*}
\set{D} = \{\tilde{y}\in\Reals\colon \tilde{y}\geq\Upsilon\}
\end{equation*}
 i.e., at every time instant $k$ the quantizer produces $1$ if \mbox{$\tilde{Y}_k\geq\Upsilon$} and $-1$ otherwise. We choose the threshold $\Upsilon$ such that the probability that the quantizer produces $-1$ given that the transmitter sends $\xi$ is equal to some arbitrary $0<\eps< 1$, i.e.,
\begin{equation}
\Upsilon = \xi -\sigma Q^{-1}(\eps)
\end{equation}
which yields
\begin{equation*}
P\bigl(Y_k=-1 \bigm| X_k=\xi\bigr) = Q\left(\frac{\xi-\Upsilon}{\sigma}\right) = \eps.
\end{equation*}
Here $Q^{-1}(\cdot)$ denotes the inverse $Q$-function. Note that this threshold is of the same form as the threshold we chose in Section~\ref{sec:mainproof} to prove Theorem~\ref{thm:main}.

The decoder guesses $\hat{M}=m$ if $Y_m=1$ and $Y_k=-1$ for $k\neq m$. If $Y_k=1$ for more than one $k$, or if $Y_k=-1$ for all $k=1,2,\ldots,\const{M}$, then the decoder declares an error.

Suppose that message $M=m$ was transmitted. The probability of an error is given by
\begin{IEEEeqnarray}{lCl}
\IEEEeqnarraymulticol{3}{l}{\Prob\bigl(\textnormal{error} \bigm| M=m\bigr)}\nonumber\\
\quad & = & \Prob\left(\left.\bigcup_{k\neq m}(Y_k=1)\cup (Y_m=-1)\right| M=m\right)\nonumber\\
& \leq & \sum_{k\neq m} P\bigl(Y_k=1\bigm| X_k=0\bigr) + P\bigl(Y_m=-1\bigm|X_m=\xi\bigr) \nonumber\\
& = & \sum_{k\neq m} P\bigl(Y_k=1\bigm|X_k=0\bigr) + \eps \nonumber\\
& = & (\const{M}-1)\,  P\bigl(Y_1=1\bigm|X_1=0\bigr) + \eps \label{eq:union}
\end{IEEEeqnarray}
where the second step follows from the union bound; the third step follows from our choice of $\Upsilon$; and the fourth step follows because the channel is memoryless which implies that the probability $\Prob(Y_k=1|X_k=0)$ does not depend on $k$. Since the RHS of \eqref{eq:union} does not depend on $m$, it follows that the probability of error
\begin{equation*}
\Prob(\hat{M}\neq M) = \frac{1}{\const{M}}\sum_{m=1}^{\const{M}} \Prob\bigl(\textnormal{error} \bigm| M=m \bigr)
\end{equation*}
is also upper-bounded by \eqref{eq:union}.

The first term on the RHS of \eqref{eq:union} can be evaluated as
\begin{IEEEeqnarray}{lCl}
\IEEEeqnarraymulticol{3}{l}{(\const{M}-1)\, P\bigl(Y_1=1\bigm|X_1=0\bigr)}\nonumber\\
\quad & = & (\const{M}-1) \,Q\left(\frac{\xi-\sigma Q^{-1}(\eps)}{\sigma}\right) \nonumber\\
& = & (\const{M}-1) \,Q\left(\frac{\sqrt{\log\const{M}}-\sigma Q^{-1}(\eps)\sqrt{\dot{R}(0)}}{\sigma\sqrt{\dot{R}(0)}}\right) \label{eq:P10}
\end{IEEEeqnarray}
where the second step follows from \eqref{eq:xi}. We continue by showing that if
\begin{equation*}
\dot{R}(0) < \frac{1}{2\sigma^2}
\end{equation*}
then, for every fixed $0<\eps<1$, the RHS of \eqref{eq:P10} tends to zero as $\const{M}$ tends to infinity.  Indeed, we have
\begin{IEEEeqnarray}{lCl}
\IEEEeqnarraymulticol{3}{l}{\lim_{\const{M}\to\infty}(\const{M}-1) \,Q\left(\frac{\sqrt{\log\const{M}}-\sigma Q^{-1}(\eps)\sqrt{\dot{R}(0)}}{\sigma\sqrt{\dot{R}(0)}}\right)}\nonumber\\
\,\,\, & \leq & \lim_{\alpha\to\infty} \exp\left(\sigma^2\dot{R}(0)\left(\alpha+Q^{-1}(\eps)\right)^2\right) Q(\alpha) \nonumber\\
& \leq & \lim_{\alpha\to\infty}\frac{1}{\sqrt{2\pi}\alpha}\exp\left(\sigma^2\dot{R}(0)\left(\alpha+Q^{-1}(\eps)\right)^2-\frac{1}{2}\alpha^2\right)\IEEEeqnarraynumspace\label{eq:Mlarge}
\end{IEEEeqnarray}
where the first step follows by upper-bounding $\const{M}-1 < \const{M}$ and by substituting
\begin{equation*}
\alpha=\frac{\sqrt{\log\const{M}}-\sigma Q^{-1}(\eps)\sqrt{\dot{R}(0)}}{\sigma \sqrt{\dot{R}(0)}};
\end{equation*}
and the second step follows from \eqref{eq:QUB}. The limit on the RHS of \eqref{eq:Mlarge} vanishes for $\dot{R}(0)<\frac{1}{2\sigma^2}$.

Combining \eqref{eq:Mlarge} with \eqref{eq:union}, we obtain that if $\dot{R}(0)<\frac{1}{2\sigma^2}$, then the probability of error tends to $\eps$ as $\const{E}$---and hence also $\const{M}=\exp\bigl(\const{E}\dot{R}(0)\bigr)$---tends to infinity. Since $\eps$ is arbitrary, it follows that the probability of error can be made arbitrarily small by choosing $\Upsilon$ sufficiently large, thus proving that the capacity per unit-energy \eqref{eq:mainCUC} is achievable by the above PPM scheme.

\section{Channel Capacity}
\label{sec:capacity}
The definition of channel capacity is analog to that of capacity per unit-energy. A rate $R(\const{P})$ is said to be achievable if for every $\eps>0$ there exists an encoder, a quantizer, and a decoder satisfying
\begin{equation*}
\frac{1}{n} \sum_{k=1}^n x_k^2 \leq \const{P}, \quad \text{for every realization of $M$}
\end{equation*}
such that
\begin{equation*}
\frac{\log\const{M}}{n} > R(\const{P}) - \eps
\end{equation*}
and such that the probability of error $\Prob\bigl(\hat{M}\neq M\bigr)$ tends to zero as $n$ tends to infinity. The capacity $C(\const{P})$ is the supremum of all achievable rates. For the above channel, the capacity is given by \cite{gallager68}
\begin{equation}
\label{eq:capacity}
C(\const{P}) = \sup I(X_1;Y_1)
\end{equation}
where the supremum is over all quantization regions $\set{D}\subset\Reals$ and over all distributions on $X_1$ satisfying $\E{X_1^2}\leq\const{P}$.

If we do not maximize over the quantization region but assume a symmetric quantizer, i.e., $\set{D}=\{\tilde{y}\in\Reals\colon \tilde{y}\geq 0\}$, then the capacity is given by  \cite[(3.4.18)]{viterbiomura79}, \cite[Th.~2]{singhdabeermadhow09_2}
\begin{equation}
C_{\text{sym}}(\const{P}) = \log 2 - H_b\biggl(Q\Bigl(\sqrt{P/\sigma^2}\Bigr)\biggr)
\end{equation}
where $H_b(\cdot)$ denotes the binary entropy function and $Q(\cdot)$ denotes the $Q$-function, see Section~\ref{sec:intro}. In this case the capacity-achieving input distribution is binary with equiprobable mass points at $\sqrt{\const{P}}$ and $-\sqrt{\const{P}}$.

To the best of our knowledge, the capacity of the above channel (maximized over the input distribution and the quantization region) as well as the capacity-achieving input distribution and the optimal quantization region are unknown. The following two propositions present results on the latter two problems.

\begin{proposition}
\label{prop:input}
The capacity-achieving input distribution is discrete with at most three mass points $\xi_{\ell}\in\Reals$, $\ell=1,2,3$ satisfying
\begin{equation}
\sum_{\ell=1}^3 p_{\ell} \xi_{\ell} = 0 \quad \text{and} \quad \sum_{\ell=1}^3 p_{\ell} \xi^2_{\ell} = \const{P}
\end{equation}
where $p_{\ell}$ denotes the probability corresponding to mass point $\xi_{\ell}$, i.e., $0\leq p_{\ell}\leq 1$, $\ell=1,2,3$ and $p_1+p_2+p_3=1$.
\end{proposition} 
\begin{proof}
Omitted.
\end{proof}
This result is consistent with \cite[Th.~1]{singhdabeermadhow09_2}, which shows that if the quantization regions of a $\const{K}$-bit quantizer partition the real line into $2^{\const{K}}$ intervals, then a discrete input distribution with not more that $2^{\const{K}}+1$ mass points achieves the capacity.

\begin{proposition}
\label{prop:quant}
A threshold quantizer is optimal, i.e., the capacity-achieving quantizer is of the form
\begin{equation}
\set{D} = \{\tilde{y}\in\Reals\colon \tilde{y}\geq\Upsilon\}, \quad \text{for some $\Upsilon\in\Reals$}
\end{equation}
where $\Upsilon$ is determined by $\const{P}$ and $\sigma^{2}$.
\end{proposition}
\begin{proof}
Omitted.
\end{proof}

Propositions~\ref{prop:input} and \ref{prop:quant} demonstrate that the capacity of the above channel is equal to the capacity of a discrete memoryless channel with input alphabet $\{\xi_1,\xi_2,\xi_3\}$, output alphabet $\{1,-1\}$, and channel law 
\begin{equation}
\label{eq:capacityW}
\Prob\bigl(Y=1\bigm|X=\xi_{\ell}\bigr)=Q\biggl(\frac{\Upsilon-\xi_{\ell}}{\sigma}\biggr), \quad \ell=1,2,3.
\end{equation}
It can be further shown that the supremum on the RHS of \eqref{eq:capacity} is achieved, so 
\begin{equation}
\label{eq:capacity2}
C(\const{P}) = \max_{\substack{(\vect{p},\bfxi)\in\set{P}(\const{P}),\\\Upsilon\in\Reals}} I\bigl(\vect{p},W(\Upsilon,\bfxi)\bigr)
\end{equation}
where $I(\vect{p},W)$ denotes the mutual information of a channel with channel law $W$ when the channel input is distributed according to $\vect{p}$; $\set{P}(\const{P})$ denotes the set of probability vectors $\vect{p}\in[0,1]^3$ and mass points $\bfxi\in\Reals^3$ satisfying
\begin{equation}
\label{eq:cond}
\sum_{\ell=1}^3 p_{\ell} = 1, \quad \sum_{\ell=1}^3 p_{\ell}\xi_{\ell} = 0, \quad \text{and} \quad \sum_{\ell=1}^3 p_{\ell} \xi_{\ell}^2 = \const{P};
\end{equation}
and $W(\Upsilon,\bfxi)$ denotes the channel law given by \eqref{eq:capacityW}. Thus, instead of maximizing the mutual information over all Borel sets $\set{D}$ and all probability distributions on $X_1$ satisfying $\E{X_1^2}\leq\const{P}$ \eqref{eq:capacity}, it suffices to maximize the mutual information over the four real numbers $\xi_1$, $\xi_2$, $\xi_3$, and $\Upsilon$ and the three-dimensional probability vector $\vect{p}$ satisfying \eqref{eq:cond}. 

\enlargethispage{-12cm}

\section{Summary and Conclusion}
\label{sec:summary}
It is well-known that quantizing the output of the discrete-time average-power-limited Gaussian channel using a symmetric one-bit quantizer reduces the capacity per unit-energy by a factor of $2/\pi$. We have shown that this loss can be fully recovered by allowing for asymmetric one-bit quantizers with corresponding asymmetric signal constellations. We have further shown that the capacity per unit-energy can be achieved by a simple PPM scheme. For this scheme, the error probability can be analyzed directly using the union bound and the upper bound \eqref{eq:QUB} on the $Q$-function. We thus need not resort to conventional methods used to prove coding theorems, such as the method of types, information-spectrum methods, or random coding exponents.

The above results demonstrate that the 2dB power loss incurred on the Gaussian channel with symmetric one-bit output quantization is not due to the hard-decision decoder, but due to the symmetric quantizer. In fact, if we employ an asymmetric quantizer, and if we use asymmetric signal constellations, then hard-decision decoding achieves the capacity per unit-energy of the Gaussian channel.

The above results also demonstrate that a threshold quantizer is asymptotically optimal as the signal-to-noise ratio tends to zero. We have further shown that this is true not only asymptotically: for a fixed signal-to-noise ratio, we have shown that, among all one-bit output quantizers, a threshold quantizer is optimal.

\section*{Acknowledgment}
The authors would like to thank P.~Sotiriadis, who sparked their interest in the problem of quantization.



\end{document}